\documentclass[10pt,twocolumn,twoside]{IEEEtran}
\usepackage[left=0.75in,right=0.75in,top=0.73in,bottom=.76in]{geometry}
\usepackage[utf8]{inputenc}
\usepackage[T1]{fontenc}
\usepackage{url}
\usepackage{ifthen}
\usepackage{cite}

\usepackage{graphicx,epstopdf,amssymb,amsthm}
\usepackage{ulem,color}
\usepackage{epstopdf}
\usepackage[utf8]{inputenc}
\usepackage[english]{babel}
\usepackage{caption}
\usepackage{algorithm}
\usepackage{algpseudocode}
\usepackage{amsmath,amsfonts,amssymb,amsthm}
\usepackage{booktabs}
\usepackage{lipsum}

\usepackage[dvipsnames]{xcolor}

\newcommand{\baike}[1]{\textcolor{black}{#1}}

\newcommand{\bshee}[1]{\textcolor{black}{#1}}
\newcommand{\bks}[1]{\textcolor{black}{#1}}

\newtheorem{theorem}{Theorem}
\newtheorem{lemma}{Lemma}
\newtheorem{proposition}{Proposition}
\newtheorem{corollary}{Corollary}

\newtheorem{remark}{Remark}

\newtheorem{definition}{Definition}

\usepackage{cite}
\usepackage[T1]{fontenc}
\usepackage{color}
\usepackage{amsmath}
\usepackage{amsthm}
\usepackage{amssymb}
\usepackage{stmaryrd}
\usepackage{graphicx}
\usepackage{esint}

\makeatletter

\begin{document}
\pagenumbering{gobble}

\title{Optimal Mitigation of SIR Epidemics Under Model Uncertainty}

\author{
Baike She, 
Shreyas Sundaram, 
and 
Philip E. Par\'{e}*
\thanks{*Baike She, Shreyas Sundaram, and Philip E. Par\'{e} are with the Elmore Family School of Electrical and Computer Engineering at Purdue University.
E-mails: \{bshe, sundara2, philpare\}@purdue.edu.
Research supported in part 
by the C3.ai Digital Transformation Institute sponsored by C3.ai Inc. and the Microsoft Corporation, and in part 
   by the National Science Foundation, grants NSF-CMMI \#1635014 
   and NSF-ECCS \#2032258.
}

}

\maketitle
\begin{abstract}
 We study the impact of \baike{ model parameter uncertainty} 
 on optimally mitigating the spread of epidemics. We capture the epidemic spreading process using a susceptible-infected-removed (SIR) epidemic model and consider testing for isolation as the control strategy. We use a testing strategy to remove (isolate) a portion of the infected population. Our goal is to maintain the daily infected population below a certain level, while minimizing the total number of tests. Distinct from existing works on leveraging control strategies in epidemic spreading, we propose a testing strategy by overestimating the seriousness of the epidemic and study the feasibility of the system under the impact of \baike{model parameter uncertainty}. Compared to the optimal testing strategy, we establish that the proposed strategy under model parameter uncertainty will flatten the curve effectively but 
require
more tests and a longer time \bshee{period}.
\end{abstract}

\section{Introduction}
Resource allocation for epidemic mitigation is of great importance for both resource and risk management during a pandemic.
In response to the ongoing COVID-19 pandemic, researchers have  studied  the  use  of  optimal  control  formulations \cite{tsay2020modeling, perkins2020optimal, acemoglu2021optimal, morris2021optimal}. 
\bshee{The authors in} \cite{tsay2020modeling} studied an “on-off” policy between
strict social distancing and \bshee{not, considering the} social and economic costs. 
\bshee{In order to study the impact of social distancing restrictions}, \cite{perkins2020optimal} calibrated epidemic models to data from the USA. In addition, \cite{acemoglu2021optimal} proposed an optimal control strategy for epidemic mitigation by combining both molecular and serology testing, \bshee{and} 
\cite{morris2021optimal} \bshee{further} discussed 
\bshee{how} leveraging optimal/near-optimal strategies 
is not robust to implementation errors. In addition to optimal control strategies, researchers leveraged model predictive control frameworks \cite{kohler2020robust,carli2020model,zino2021analysis,she2022mpcepi}, and other  strategies \cite{khadilkar2020optimising,scarabaggio2021nonpharmaceutical}  to generate optimal/sub-optimal policies for epidemic mitigation. For instance, \cite{scarabaggio2021nonpharmaceutical} exploited the structure of the transmission networks to determine vaccination targets, while 
\cite{bastani2021efficient} formulated the \baike{COVID-19 mitigation problem} 
\bshee{using} a reinforcement learning framework.
Other works considering epidemic control and resource allocation include \cite{bloem2009optimal, nowzari2016epidemics, di2017optimal,sharomi2017optimal, di2019optimal,liu2019bivirus,dangerfield2019resource,preciado2014epidemic_optimal,han2015data}.


The aforementioned research was established upon the prior knowledge of the epidemic model parameters. Nevertheless, works regarding \bshee{real-time epidemic modeling} and prediction \cite{chowell2017fitting, baker2018mechanistic,wilke2020predicting} \bshee{have shown} that it is difficult to predict the \bshee{behavior of} epidemic spreading processes. Hence, \bshee{obtaining} accurate epidemic spreading parameters is challenging  \baike{when formulating} real-time epidemic modeling and control problems. In this work, we 
\baike{tackle optimal epidemic control problems under the impact of parameter uncertainties.}
We aim to modify the optimal epidemic mitigation strategy in \cite{casella2020can} by leveraging a range of known model parameters generated by epidemic parameter learning processes 
instead of accurate model parameters.
We consider a testing-for-isolation strategy \cite{casella2020can},
which 
removes
the infected  population from the infected group through
uniform random sampling, \baike{i.e., the control input variable}. Similar to 
vaccination strategies that remove the susceptible population from the mixed group \cite{grundel2021coordinate}, the testing-for-isolation strategy is another widely adopted method \bshee{for} epidemic mitigation\cite{casella2020can,acemoglu2021optimal,Purdue_testing}.

Our main contribution is to propose a testing strategy for epidemic mitigation under the impact of \baike{model uncertainties introduced by real-time epidemic modeling \bshee{parameter estimating,} and state estimation}. 
Specifically, we bridge the gap between parameter estimation for epidemic spreading processes and theoretical analysis \bshee{of} optimal control strategies \bshee{for} epidemic mitigation. Assuming the range of the model parameters and states are \bshee{obtained} by any given method, 
we adapt testing-for-isolation strategies  \cite{acemoglu2021optimal,casella2020can} to study the \baike{additional control cost} of the \baike{parameter and state uncertainties} on the proposed optimal testing policy \cite{casella2020can}. We propose a testing strategy by overestimating the seriousness of the epidemic to adapt the optimal testing policy under the ranges of the \bshee{obtained} parameters and states to guarantee the system feasibility. Further, by comparing the testing cost of the proposed testing strategy with the optimal testing policy, we conclude \bshee{that} the proposed testing strategy under the parameter learning and state estimation processes can flatten the curve effectively, but will cost more tests and time.

The paper is organized as follows. In Section~\ref{section2}, we introduce the optimal epidemic mitigation problem and the goal of this work. In Section~\ref{section3}, we propose a testing strategy to study the feasibility of the control problem under the \baike{parameter and state uncertainties}. We characterize the \baike{control} cost via comparison with the optimal testing strategy \baike{generated under accurate models and states}. In Section~\ref{section4}, we illustrate the proposed control strategy through simulations. Section~\ref{section5} presents the conclusions and future work.
\section{Problem Formulation}
\label{section2}
In this section, we introduce the epidemic spreading model and formulate the optimal epidemic resource allocation for mitigation problem.
Our goal is to propose a potential way for policy-makers to implement a feedback testing strategy to mitigate an epidemic through \bshee{estimated} parameters and states.
As illustrated by the arrows from the top and middle blocks to the bottom block in Fig. \ref{fig_framework}, 
we leverage the  model parameters and epidemic states \baike{with uncertainties} to study the control policy.
\begin{figure}
  \begin{center}
    \includegraphics[ trim = 1cm 0.5cm 0.5cm 0.5cm, clip, width=\columnwidth]{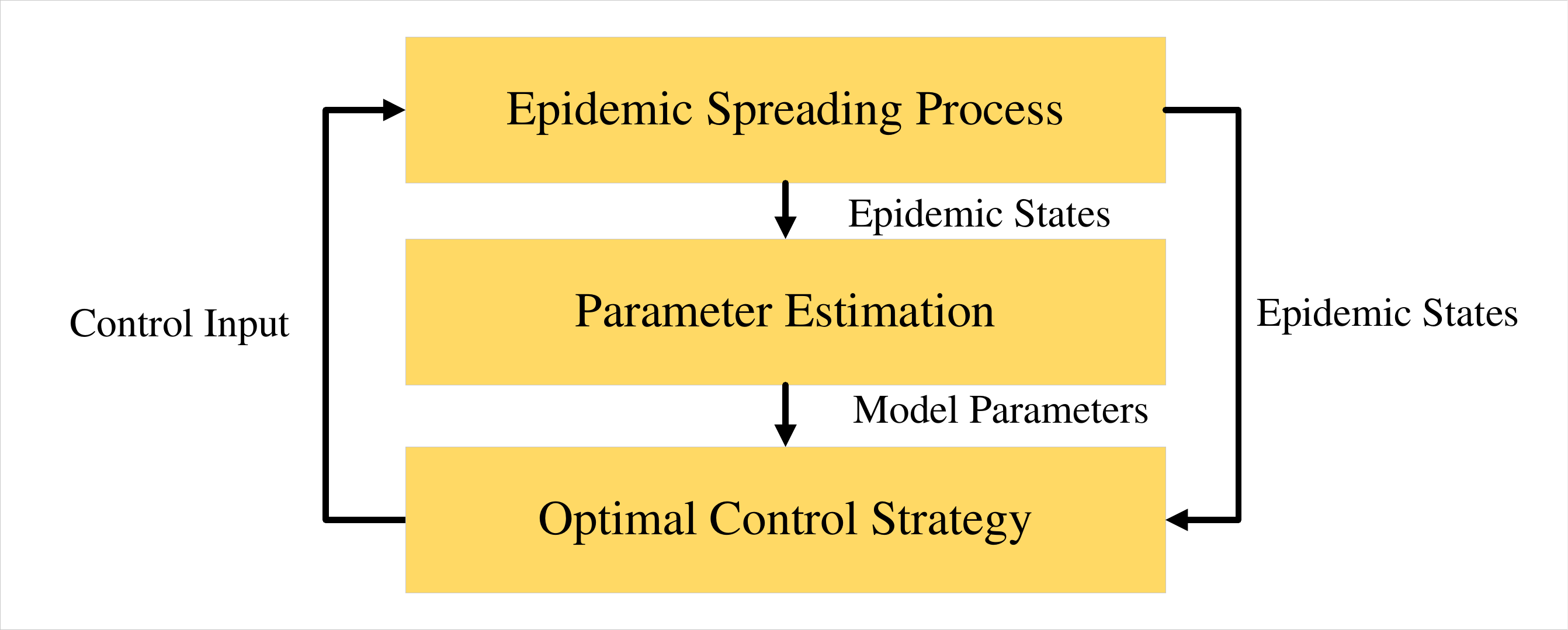}
  \end{center}
  \caption{Control Framework}
  \label{fig_framework}
\end{figure}
\subsection{Testing for Mitigation}
In this subsection, we  present the model for the epidemic control problem. We consider the following closed-loop Susceptible-Infected-Recovered/Removed (SIR) model:
\begin{subequations}\label{Eq: Con_Dynamics}
\begin{alignat}{3}
 \frac{dS(t)}{dt} &= -\beta S(t)I(t), \label{eq:S_u}\\
   \frac{dI(t)}{dt} &= \beta S(t)I(t) -(\gamma+u(t)) I(t),\label{eq:I_u}\\
    \frac{dR(t)}{dt} &= (\gamma+u(t)) I(t).\label{eq:R_u}
\end{alignat}
\end{subequations}
The parameters $\beta$ and $\gamma$ represent the time-invariant transmission rate and removal rate, respectively\bshee{, and the control input $u(t)$ captures testing strategies that isolate/remove $u(t)\times100\%$ of the detected infected population from the infected group, represented by $u(t)I(t)$}. In this work, we assume the removal rate captures any processes that separate the detected infected group from the whole population, which include the recovery process, hospitalization, deaths, etc. 
We define mitigation as
maintaining the infection level  under a certain threshold through control strategies.
Note that when $u(t)=0$, the system in \eqref{Eq: Con_Dynamics} becomes the classic SIR model \cite{kermack1927_sir}. 

\subsection{{Optimal Testing Problem}}
In this subsection, we introduce the  optimal control framework.
Consider the system formulated in \eqref{Eq: Con_Dynamics}.  The goal for the epidemic mitigation problem is to optimally allocate the testing resources during the pandemic such that the daily infected population is maintained at/below the desired infection threshold. 
In this work, we  consider mitigating the epidemic by minimizing the total number of tests during the epidemic through the \bshee{following} cost function 
\begin{equation} 
\label{eq:cost_function}
J(u(t))=\int_{0}^{+\infty}u(t)dt.
\end{equation}
In order to obtain the \bshee{testing-for-isolation strategy} that minimizes the total
number of tests needed during the epidemic spreading process while ensuring that the fraction of infected 
individuals  
remains below a desired threshold, we formulate the following optimization problem, 
\begin{subequations}\label{eq:prob}
\begin{align}
&\min_{u(t), 0\leq t \leq +\infty}  \, \, J(u(t)) \\
&\text{s.t.}  \, \, 
\dot{\boldsymbol{x}}(t)=f(\boldsymbol{x}(t),u(t)),  \\ 
\label{eq:constraints}
&0 \leq I(t) \leq \bar{I}, 
\underline{u} \leq {u}(t) \leq \bar{u}, \forall t \in [0, +\infty),
\end{align}
\end{subequations}
where $\dot{\boldsymbol{x}}(t)=f(\boldsymbol{x}(t),u(t))$ denotes the closed-loop dynamics in \eqref{Eq: Con_Dynamics}. The state constraint $\bar{I}$ describes the \textit{infection threshold} for the fraction of the infected undetected
population. In addition, the control input constraints $\underline{u}$ and $\bar{u}$ define the lower and upper bounds on the testing rates, respectively. 
\subsection{Goals}
In this work, we assume the ranges of the model parameters and states in \eqref{Eq: Con_Dynamics}  are given via potential existing \baike{real-time modeling} and estimation techniques at any given time $t\geq0$. 
We use $S(t)$, $I(t)$, $R(t)$ $\forall t\geq 0$ to denote the true susceptible, infected, and \baike{removal} states, respectively, while $\hat S(t)$, $\hat I(t)$, $\hat R(t)$ $\forall t\geq 0$ 
represent the \baike{corresponding} estimated states. 
Distinct from the true model parameters $\beta$ and $\gamma$, we use $\hat{\beta}(t)$ and $\hat\gamma(t)$ $\forall t\geq 0$ to represent the \baike{estimated} parameters at any given time $t\geq0$. In addition, 
we assume $\hat{\beta}(t)
,
\beta\in [\hat{\beta}_{\min}(t),\hat{\beta}_{\max}(t)]$; $\hat{\gamma}(t),
\gamma \in [\hat{\gamma}_{\min}(t),\hat{\gamma}_{\max}(t)]$; 
$\hat S(t),
S(t) \in[\hat{S}_{\min}(t),\hat{S}_{\max}(t)]$; and $\hat I(t),
I(t) \in[\hat{I}_{\min}(t),\hat{I}_{\max}(t)]$ $\forall t\geq 0$.
Moreover, we use $S^*(t)$, $I^*(t)$, $R^*(t)$ to represent the true states under the optimal control strategy $u^*(t)$ $\forall t\geq 0$ for the problem defined in \eqref{eq:prob}.

We  focus on the theoretical analysis of the optimal control for the epidemic mitigation problem defined in \eqref{eq:prob}, under the impact of the \baike{parameter and state uncertainties}. We study optimal control strategies of  \eqref{eq:prob} in order to
propose a testing strategy by leveraging the estimated model parameters and states. 
We explore the \textit{additional control cost} by comparing the total number of tests generated from the proposed control strategy with the tests under the optimal testing strategy. We aim to show the effectiveness of the proposed testing strategy through overestimating the seriousness of the epidemic under the \baike{existence of parameter and state uncertainties}.

\section{Testing for Epidemic Mitigation }
\label{section3}
\bshee{We explore the feasibility and additional cost of the optimal control framework  proposed in Fig. \ref{fig_framework} in this section.}
\subsection{Feasibility and the Optimal Testing Strategy}
We first study the optimal control framework in \eqref{eq:prob} under accurate model parameters and states. 
Let $t=0$ denote the very beginning
of an epidemic, and $t_p$ denote the time when the infection state reaches the peak value during the epidemic spreading process, i.e., $I(t_p)\geq I(t)$ $\forall t\geq 0$. The following lemma characterizes the peak value $I(t_p)$ 
in $\eqref{Eq: Con_Dynamics}$.

\begin{lemma}
\label{lem:Ip}
Starting from $\boldsymbol{x}(t_a)=[S(t_a)\quad I(t_a)\quad R(t_a)]^{\top}$ and $u(t_a)=\underline{u}$ at time $t_a<t_p$, if the system in \eqref{Eq: Con_Dynamics} under the fixed control input $u(t)=\underline{u}$ reaches a peak infection value $I(t_p)$,  we have $I(t_p) = \rho (\ln\rho -1- \ln S(t_a))+S(t_a)+I(t_a)$, where $\rho=\frac{\gamma+\underline{u}}{\beta}$.
\end{lemma}
\begin{proof}
Consider \eqref{Eq: Con_Dynamics} $\forall t\geq t_a$, dividing \eqref{eq:I_u} by  \eqref{eq:S_u} \bshee{gives} $dI(t)/dS(t) = (\gamma+u(t))/\beta S(t)-1$.
Then, we integrate the equation with respect to $S(t)$ and apply the initial conditions $\boldsymbol{x}(t_a)=[S(t_a)\quad I(t_a)\quad R(t_a)]^{\top}$ and $u(t_a)=\underline{u}$. Then by fixing $u(t)=\underline{u}$, we obtain
\begin{equation*}
    I(t) = \frac{\gamma+\underline{u}}{\beta} \ln S(t) -S(t)-\frac{\gamma+\underline{u}}{\beta} \ln S(t_a)+S(t_a)+I(t_a), 
\end{equation*}
$\forall t\geq t_a$.
From 
\eqref{eq:I_u},
the infected population at $t_p$ satisfies $\frac{dI(t_p)}{dt}=\beta S(t_p)I(t_p)-(\gamma+\underline{u})I(t_p)=0$, and $I(t_p)\neq0$. Hence, we have $S(t_p)=\frac{\gamma+\underline{u}}{\beta}$ at $t_p$. 
By evaluating $I(t)$ at $t_p$ and substituting in $S(t_p)=\frac{\gamma+\underline{u}}{\beta}=\rho$, we have
$I(t_p) = \rho \ln\rho -\rho-\rho \ln S(t_a)+S(t_a)+I(t_a).$
\end{proof}
Lemma \ref{lem:Ip}
\bshee{calculates}
the peak infection value $I(t_p)$ from any initial condition $\boldsymbol{x}(t_a)$ and $u(t_a)$ before $t_p$, under the fixed control input $u(t)=\underline{u}$ $\forall t\geq0$. 
Note that if $\underline{u}=\bar{u}=0$, Lemma \ref{lem:Ip} characterizes the peak infection value for the classic SIR model.
\begin{corollary}
\label{Lem: cor}
Assume the closed-loop system in \eqref{Eq: Con_Dynamics} starts from $\boldsymbol{x}(t_a)=[S(t_a)\quad I(t_a)\quad R(t_a)]^{\top}$ and $u(t_a)=\underline{u}$ at time $t_a$.
If $\exists$ $t_{p}$ s.t. $I(t_{p})\geq I(t)$ $\forall t\geq t_a$, 
the peak infection value $I(t_{p})$ will increase as $\beta$ increases; decrease as $\gamma$ increases; and decrease as $\underline{u}$ increases.
\end{corollary}
\begin{proof}
Consider $I(t_{p})$ as a function of $\rho$ in Lemma \ref{lem:Ip}. Since $I(t_p)$ is the peak infection value during the epidemic spreading process, and $t_p>t_a$, then we have $\frac{dI(t)}{dt}>0$ $\forall t\in [t_a, t_p)$. 
From \eqref{eq:I_u}, we have $\frac{\gamma+\underline{u}}{\beta S(t)}< 1$, $\forall t\in[t_a, t_{p})$. 
Define function $g(\rho)=\rho \ln\rho-\rho-\rho \ln S(t_a)$, $\frac{\rho}{S(t)}\in  (0,1)$, $\forall t\in[t_a, t_{p})$. We can obtain that the first derivative
$g'(\rho)=\ln\frac{\rho}{S(t_a)} 
<0$, since $\frac{\rho}{S(t_a)}\in (0,1)$. Therefore,
$g(\rho)$ is monotonically decreasing with respect to $\rho$, and thus $g(\rho)$ is monotonically decreasing with respect to $\underline{u}$. Furthermore, $I(t_p)$ is monotonically decreasing with respect to $\gamma$ and $\underline{u}$, and monotonically increasing with respect to $\beta$, $\forall t\in [t_a, t_p)$. Hence, we complete the proof.
\end{proof}

Corollary \ref{Lem: cor} implies that, 
under the same initial conditions, 
the peak infection value $I(t_p)$ will decrease with higher $\beta$ and/or lower $\gamma$. 
Further, 
Corollary~\ref{Lem: cor} states that increasing the lower bound on the testing rate $\underline{u}$ will lower the peak infection value. Hence, if $\underline{u}$ in $\eqref{eq:prob}$ is sufficiently high, such that $I(t_p)\leq \bar{I}$ when $u(t)=\underline{u}$ $\forall t\geq 0$,
the optimal control strategy will be $u(t)=\underline{u}$ $\forall t\geq 0$. 

\begin{corollary}[Optimal Testing Strategy 1]
\label{lem:opt1}
The optimal testing strategy for the problem in \eqref{eq:prob} is $u^*(t)=\underline{u}$  $\forall t\geq 0$, if $I^*(t_p) = \rho (\ln\rho -1- \ln S^*(0))+S^*(0)+I^*(0)\leq \bar{I}$.
\end{corollary}
Corollary~\ref{lem:opt1} is a direct result from Lemma \ref{lem:Ip} and Corollary~\ref{Lem: cor}, thus the proof is omitted. For the optimal control problem in \eqref{eq:prob}, if there is no risk for the infection state to exceed the infection threshold $\bar{I}$, 
maintaining the testing at $\underline{u}$ is the best way to reduce the cost. For the control framework in \eqref{eq:prob}, 
we consider the case \bshee{when} $I(t_p)> \bar{I}$ under $u(t)=\underline{u}$, $\forall t\geq 0$, 
and develop the following theorem to study the feasibility of the framework in \eqref{eq:prob}.

\begin{theorem}
\label{Thm:feasibility}
Starting from $t=t_a\geq 0$, if $\exists t_b\geq t_a$ s.t. $I(t_b)=\bar{I}$ for the first time, then
the control framework in \eqref{eq:prob} is feasible if and only if $\exists u(t_b)\in(\underline{u}, \bar{u}]$ s.t. $u(t_b)=\beta S(t_b)-\gamma$.
\end{theorem}
\begin{proof}
Consider the system in \eqref{Eq: Con_Dynamics} before reaching $t_b$, we have $I(t)<\bar{I}$, $\forall t\in [0, t_b)$. 
Hence, the system is feasible $\forall t\in [0, t_b)$. Then we study the system starting from $t_b$.\\
$\Longleftarrow$: Under the condition that $\exists t_b\geq t_a$ s.t. $I(t_b)=\bar{I}$ for the first time, if $\exists u(t_b)\in(\underline{u}, \bar{u}]$ s.t. $u(t_b)=\beta S(t_b)-\gamma$, from 
\eqref{eq:I_u},
we have $\frac{dI(t_b)}{dt}=0$. 
Furthermore, since $S(t)$ is strictly monotonically decreasing unless $S(t)=0$ and/or $I(t)=0$  $\forall t\geq0$, we can always find a $u(t)\in [u(t_b), \bar{u}]$, such that
$u(t)\geq \beta S(t)-\gamma$ $\forall t\geq t_b$. 
From \eqref{eq:I_u}, there  always exists a $u(t)\in  [u(t_b), \bar{u}]$ such that $\frac{dI(t)}{dt}\leq 0$ $\forall t\geq t_b$, which guarantees
$I(t)\leq \bar{I}$ $\forall t\geq t_b$. Therefore, the control framework in \eqref{eq:prob} is feasible.\\
$\Longrightarrow$: Starting from $t=t_a\geq 0$, $\exists t_b>t_a$ s.t. $I(t_b)=\bar{I}$ for the first time. If
the system is feasible, $I(t)$ must stop increasing at $t_b$. Hence, $\frac{dI(t_b)}{dt}\leq 0$ indicates that there must exist $u(t_b)\in (\underline{u},\bar{u}]$ such that $u(t_b)=\beta S(t_b)-\gamma$, which completes the proof.
\end{proof}
In this work, we study the case that satisfies Theorem \ref{Thm:feasibility}:
the upper bound on the testing rate $\bar{u}$ is \bshee{sufficiently large} such that we can always find a $u(t_b)\in(\underline{u}, \bar{u}]$, to satisfy $u(t_b)= \beta S(t_b)-\gamma$. Under such condition, the optimal testing strategy is given by the following proposition, \bks{where $a^*$, $a\in \{S,I,R,t_b,t_h\}$ represents the state or the time step of the system in \eqref{Eq: Con_Dynamics}
under the optimal control strategy $u^*(t)$.
Note that $t^*_b$ is the time step when $I^*(t)$ $t\geq 0$ reaches $\bar{I}$ under the optimal testing strategy $u^*(t)$ for the first time.
In addition, $t^*_h$ is time step  when the epidemic reaches herd immunity  under the optimal testing strategy $u^*(t)$ for the first time, i.e., $\frac{dI(t^*_h)}{dt}=(\beta S(t^*_h)-(\gamma+\underline{u}))I(t^*_h)=0$. Furthermore, we have $\frac{dI(t^*_h)}{dt}\leq 0$, $\forall u(t)\in[\underline{u}, \bar{u}]$, $\forall t\geq t^*_h$.}  
\begin{proposition}[Optimal Testing Strategy 2]
\label{Prop:policy}
\cite[Theorem 1]{acemoglu2021optimal}
The optimal testing strategy for the problem in \eqref{eq:prob} can be cast into three stages:
\begin{enumerate}
    \item At the early stage of the epidemic,  when $I^*(t)<\bar{I}$, $\forall t\in[0, t^*_b)$, $u^*(t)=\underline{u}$;
    \item During the outbreak, starting from $I^*(t^*_b)=\bar{I}$, $\forall t\in[t^*_b, t^*_h)$, $u^*(t)=\beta S^*(t)-\gamma$;
    \item When the epidemic reaches herd immunity at $t^*_h$, i.e., $\beta S^*(t^*_h)=\gamma+\underline{u}$, $\forall t\geq t^*_h$, $u(t)=\underline{u}$.
\end{enumerate}
\end{proposition}

The proof of Proposition \ref{Prop:policy} is the same as the proof of \cite[Theorem 1]{acemoglu2021optimal}, although the lower bound on the testing rate is $\underline{u}=0$ in \cite[Theorem 1]{acemoglu2021optimal}. Proposition \ref{Prop:policy} separates the testing strategy into three stages via considering the first time when the infection state reaches $\bar{I}$, i.e., $t^*_b$, and the herd immunity time step $t^*_h$ as the switching time steps. In the following subsection, we aim to explore testing strategies under the guidance of the optimal testing strategy in Proposition \ref{Prop:policy}, \baike{with parameter and state uncertainties.}

\subsection{Testing Strategy under \baike{Uncertainties}}
In this subsection, we propose a testing strategy for the problem in \eqref{eq:prob} with \baike{parameter and state uncertainties captured by}
the ranges given in \bshee{Section \ref{section2}}. \bshee{Recall that we define $\hat X(t)$, $X\in \{S,I,R\}$, $\forall t\geq 0$ as the estimated states. We use $\hat{t}_b$ to denote the time step when the overestimated state $\hat I_{\max} (t)$ reaches the infection threshold $\bar{I}$ for the first time. In addition, we use $\hat{t}_h$ to represent the time step  when  $\hat{\beta}_{\max}(\hat{t}_h) \hat{S}_{\max}(\hat{t}_h)=\hat{\gamma}_{\min}(\hat{t}_h)+\underline{u}$ for the first time, i.e., the computed herd immunity time step by overestimating the epidemic states and spreading parameters. We use $\hat{u}(t)$ $\forall t\geq 0$ to represent the generated testing strategy by leveraging the overestimated epidemic spreading process and the corresponding computed time steps $\hat t_h$ and $\hat t_b$.}

\begin{definition}[Testing Strategy under \baike{Uncertainties}]
The testing strategy for the problem in \eqref{eq:prob} follows the rules:
\label{def:optimal_p}
\begin{enumerate}
     \item At the early stage of the epidemic, \bshee{when the overestimated infection state is smaller than the infection threshold $\bar{I}$, the testing strategy is given by $\hat u(t)=\underline{u}$, $\forall t\in [0, \hat{t}_b)$; 
    \item From the time step $\hat{t}_b$ to the computed herd immunity time step $\hat{t}_h$,
    the testing strategy is given by 
    $\hat u(t)=\hat{\beta}_{\max}(t)\hat{S}_{\max}(t)-\hat{\gamma}_{\min}(t)$, $\forall t\in[\hat{t}_b, \hat{t}_h)$;}
    \item Starting from the computed herd immunity time step $\hat{t}_h$, the testing strategy is given by  $\hat u(t)=\underline{u}$, $\forall t\geq \hat{t}_h$.
\end{enumerate}
\end{definition}
Definition \ref{def:optimal_p} modifies the optimal testing strategy in Proposition \ref{Prop:policy} by proposing a testing policy under the given ranges of estimated parameters and states. Definition \ref{def:optimal_p} implies that without accurate model parameters and states, \baike{if we know the range of the parameters and states,} the testing strategy will always assume the worst case 
scenario at any given time step to generate the testing policy, i.e., to overestimate the seriousness of the epidemic. 

We discuss the feasibility of the system in \eqref{Eq: Con_Dynamics} under the proposed testing strategy in Definition~\ref{def:optimal_p} by first studying the  situation where 
$\hat{\beta}_{\max}(t)=\hat{\beta}_{\max}\geq\beta$, $\hat{\gamma}_{\min}(t)=\hat{\gamma}_{\min}\leq\gamma$, 
$\forall t\geq0$, and $\hat S(t) \in[S(t) ,\hat{S}_{\max}(t)]$, $\hat I(t) \in[\hat I(t),\hat{I}_{\max}(t)]$ $\forall t\geq 0$. 
This case assumes the estimated ranges of the parameters are time-invariant. 
Recall that $S^*(t)$, $I^*(t)$, $R^*(t)$ denote the system's trajectories under the optimal testing strategy $u^*(t)$, $\forall t\geq0$. Similar to the definitions of $t^*_b$ and $t^*_h$, we define $\hat t_b$  and $\hat t_h$ as the time steps when $ \hat {I}(\hat t_b) = \bar{I}$ for the first time and $\underline{u} = \hat{\beta}_{\max}\hat{S}(\hat t_h)-\hat{\gamma}_{\min}$ for the first time, respectively. We 
plot
both trajectories of the system under the optimal testing strategy $u^*(t)$ and the strategy $\hat u(t)$ from Definition~\ref{def:optimal_p} in Fig.~\ref{fig:example}, in order to better explain  $t^*_b$, $t^*_h$, $\hat t_b$, and $\hat t_h$. \bks{ Fig.~\ref{fig:example} compares the behavior of the epidemic under the testing strategy in Definition~\ref{def:optimal_p} when overestimating the spreading parameters, with the behavior of the epidemic under the optimal testing strategy in Proposition~\ref{Prop:policy} when the true spreading parameters are known. Consider an epidemic spreading process with $\beta=0.016$ and $\gamma=0.033$. The infection threshold is set as $\bar{I}=0.01$. The lower bound on the testing rate is $\bar{u}=0.03$. We use $S^*(t)$, $I^*(t)$, and $R^*(t)$ $t\geq0$ to represent the states generated by $u^*(t)$ following
the Optimal Testing Strategy~1 in Proposition~\ref{Prop:policy}. We use $S(t)$,
$I(t)$, and $R(t)$ $t\geq0$ to denote the true states  generated by $\hat u(t)$, when implementing 
the testing strategy given in Definition~\ref{def:optimal_p}
and leveraging the overestimated spreading parameters $\hat{\beta}(t)=1.05\beta$ and  $\hat{\gamma}(t)=0.95\gamma$ $\forall t\geq0$,
and noisy states $\hat{S}(t)$ and $\hat{I}(t)$ $\forall t \geq 0$.}
From Definition~\ref{def:optimal_p}, we \bshee{will leverage Fig.~\ref{fig:example} to illustrate} the following result.
\begin{figure*}
    \centering
  \includegraphics[width=1.0\textwidth]{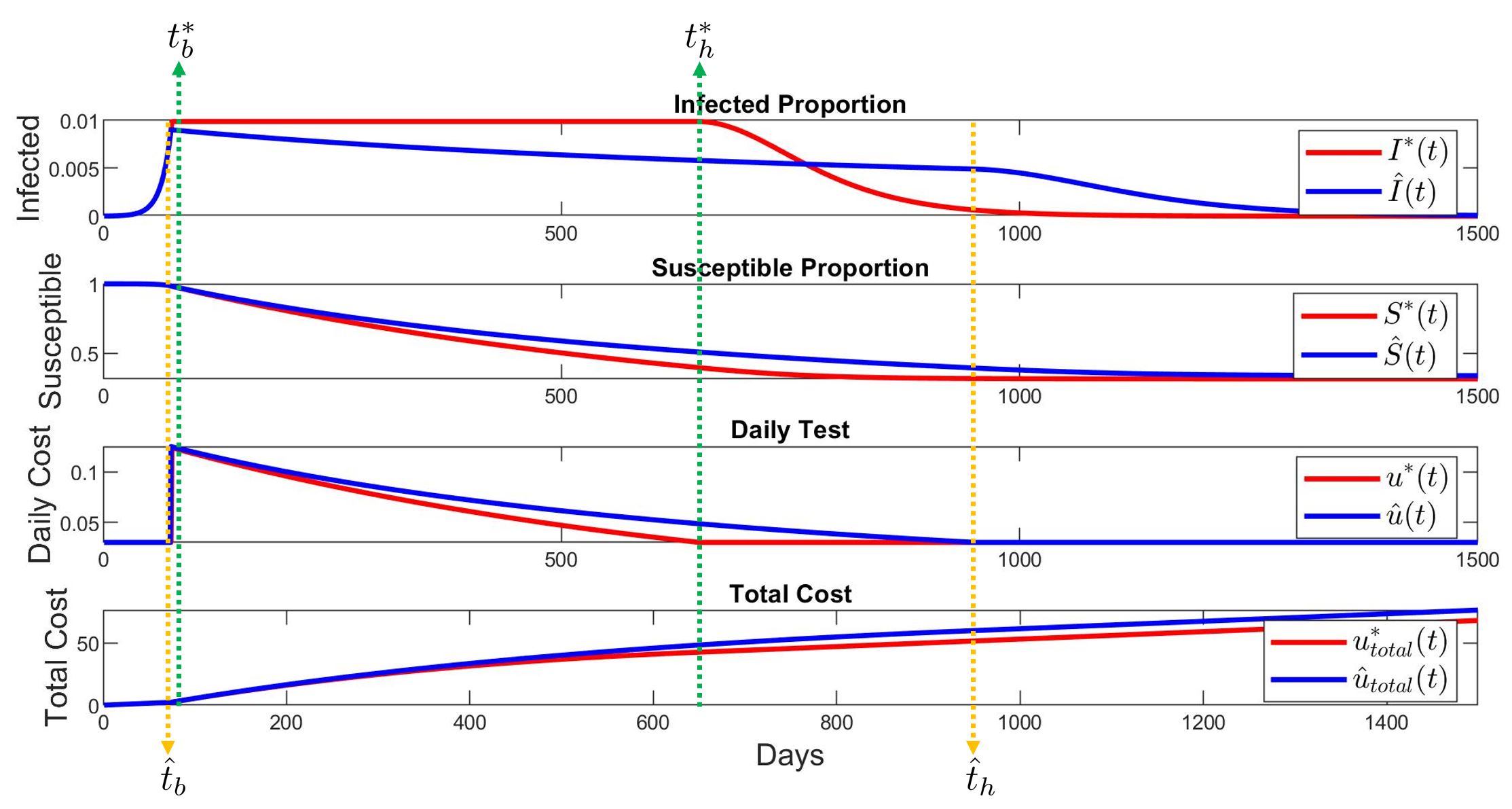}
    \caption{Comparison of Lemma \ref{lem:inaccu} with the Optimal Testing Strategy.}
    \label{fig:example}
\end{figure*}
\begin{lemma}
\label{lem:inaccu}
When $\hat{\beta}(t)=\hat{\beta}_{\max}\geq\beta$, $\hat{\gamma}(t)=\hat{\gamma}_{\min}\leq\gamma$, $\hat S(t) \in[S(t) ,\hat{S}_{\max}(t)]$, $\hat I(t) \in[\hat I(t),\hat{I}_{\max}(t)]$,  $\forall t\geq 0$, the system in \eqref{Eq: Con_Dynamics} under the control strategy $\hat u(t)$ generated by leveraging  $\hat{\beta}_{\max}$, $\hat{\gamma}_{\min}$, $\hat{S}(t)$, $\hat I(t)$ $\forall t\geq0$,
from Definition~\ref{def:optimal_p} is feasible. The control strategy satisfies $\hat u(t)\geq u^*(t)$, $\forall t\geq0$,
\end{lemma}

\begin{proof}
We compare $u^*(t)$ and $\hat{u}(t)$ by considering $t\in [0, \hat t_b)\cup [\hat t_b, t^*_b]\cup [t^*_b, t^*_h]\cup(t^*_h, \hat t_h] \cup (\hat t_h, +\infty)$, where the chronological order will be demonstrated within the context.
First, we show the system in \eqref{Eq: Con_Dynamics} under the testing policy $\hat u(t)$ $\forall t\geq 0$ is feasible. We analyze the testing strategy by considering three main testing stages. Recall that the control framework first switches its testing policy when $\hat I(\hat t_b)=\bar{I}$ ($\hat{t}_b$ is the first time when $\hat{I}(t)$ reaches $\bar{I}$, as shown in the \bshee{top plot of} Fig.~\ref{fig:example}). 
Since $\hat I(t)\geq I(t)$, $\forall t\geq 0$, we have 
$I^*(\hat{t}_b)=I(\hat{t}_b)\leq \hat I(\hat{t}_b)=\bar{I}$. Hence, compared to using the optimal testing policy $u^*(t)$ $\forall t\in [0, t^*_b]$, the system, by leveraging larger estimated infection states, will start to raise the testing rate away from the lower bound earlier, i.e., at $\hat{t}_b$. 
Hence, we have $\hat{t}_b\leq{t}^*_b$, as illustrated in Fig. \ref{fig:example}. In addition,
at the early stage of the epidemic, when $\hat I(t)<\bar{I}$, $\forall t\in[0, \hat t_b)$, we have $\hat u(t)=u^*(t)=\underline{u}$, $\forall t\in[0, \hat t_b)$. Then we consider the time step when $\hat I(\hat t_b)=\bar{I}$. From
Definition \ref{def:optimal_p}, we have
$(\beta S(t)-(\gamma+\hat u(t))\leq (\hat{\beta}_{\max} \hat{S}(t)-(\hat{\gamma}_{\min}+\hat u(t)) =0$. Thus $\frac{dI(t)}{dt}\leq 0$,
$\forall t\in [\hat t_b, \hat{t}_h]$, where
$\hat{t}_h$ is the computed herd immunity time step under the condition that $\hat{S}(\hat{t}_h)\hat{\beta}_{\max}-\hat{\gamma}_{\min}=\underline{u}$ (shown in Fig. \ref{fig:example}). Hence, the infection state $I(t)$ is non-increasing under $\hat u(t)$, and  $I(t)\leq \bar I$,
$\forall t\in [\hat {t}_b, \hat{t}_h]$. Lastly, after reaching the computed herd immunity time step $\hat t_h$, from Definition \ref{def:optimal_p}, we have $\hat u(t)=\underline{u}$ and $(\beta S(t)-(\gamma+\underline{u}))\leq (\hat{\beta}_{\max}\hat{S}(t)-(\hat{\gamma}_{\min}+\underline{u}))\leq 0$, $\forall t\geq \hat{t}_h$. Therefore,
$I(t)$ $\forall t\geq \hat{t}_h$ will monotonically decrease, and thus cannot reach $\bar{I}$ again. The trajectories of the optimal states under $u^*(t)$ and the true states under the testing strategy $\hat{u}(t)$ $\forall t\geq 0$ are shown in Fig. \ref{fig:example}. 
In summary, starting from $t=0$, $I(t)$ cannot exceed $\bar{I}$ under the given control policy $\hat u(t)$ $\forall t\geq 0$, which completes the proof of feasibility. 

Now we compare $\hat u(t)$ and $u^*(t)$.  Recall at the early stage of the epidemic, when $\hat I(t)<\bar{I}$, $\forall t\in[0, \hat{t}_b)$, $\hat u(t)=u^*(t)=\underline{u}$.
Starting from $\hat t_b$, we have $\hat u(t)=\hat S(t)\hat\beta_{\max}-\hat\gamma_{\min}\geq\underline{u}=u^*(t)$ $\forall t\in[\hat{t}_b, t^*_b]$. Note that $\hat u(t)$ is not the optimal control strategy (but a strategy that ensures the system is feasible) for the problem defined in \eqref{eq:prob}. Moreover, \cite[Lemma~5]{acemoglu2021optimal} states that, among all the feasible frameworks, the system in \eqref{Eq: Con_Dynamics} reaches the herd immunity time step $t^*_h$ the fastest, under the optimal testing strategy $u^*(t)$. Hence, we have $\hat{t}_h\geq t_h\geq t^*_h$. Recall that $\hat S(t)$ and $S(t)$ $\forall t\geq 0$ are the estimated susceptible state and the corresponding true state under the control policy from Lemma \ref{lem:inaccu}, respectively.
In addition, $t_h$ and $\hat{t}_h$  are the time steps when $S(t_h)\beta-\gamma=\underline{u}$ and $\hat{S}(\hat{t}_h)\beta_{\max}-\gamma_{\min}=\underline{u}$ under the control policy $\hat u(t)$, respectively. The inequality $\hat{t}_h\geq t_h$ implies that when $S(t_h)\beta-\gamma=\underline{u}$, the estimated parameters and states still satisfy $\hat S(t_h)\hat\beta_{\max}-\hat\gamma_{\min}\geq\underline{u}$. Thus, 
it will take longer for the system to reach the estimated herd immunity time step $\hat t_h$. 
Further, the system in \eqref{Eq: Con_Dynamics} under the optimal control policy $u^*(t)$ will reach the herd immunity time step $t^*_h$ faster (or equal to) the system in \eqref{Eq: Con_Dynamics} under $\hat{u}(t)$ (i.e., the estimated herd immunity time $\hat{t}_h$).
From Proposition \ref{Prop:policy} and Definition~\ref{def:optimal_p},  $u^*(t)=\underline{u}$, $\forall t\geq t^*_h$, and $\hat u(t)=\hat{S}(t)\hat{\beta}_{\max}-\hat{\gamma}_{\min}\geq \underline{u}$, $\forall t\in [t^*_h, \hat{t}_h]$. In addition, we have $\hat u(t)=\underline{u}$, $\forall t\geq \hat{t}_h$, which  leads to $\underline{u}=u^*(t)\leq \hat u(t)$ $\forall t\geq t^*_h$, eventually.

Lastly, we analyze both testing policies when $t\in [t^*_b, t^*_h]$. Following the discussion from the feasibility and the fact that the optimal control strategy $u^*(t)$ maintains $I^*(t)=\bar{I}$ $\forall t\in [t^*_b, t^*_h]$, we have $I(t)\leq\bar{I}=I^*(t)$,  $\forall t\in [t^*_b, t^*_h]$. Hence, from the integration of \eqref{eq:S_u} (dividing $S(t)$ on both sides):
$log (S(t))=log (S(t^*_b))-\int_{t^*_b}^{t}(\beta I(\tau)) d\tau$, if $ I(t)\leq\bar{I}=I^*(t)$,  $\forall t\in [t^*_b, t^*_h]$, then $\hat S(t)\geq S(t)\geq S^*(t)$, $\forall t\in [t^*_b, t^*_h]$ (note that $S(t^*_b)\geq S^*(t^*_b)$). 
From 
the fact that $\hat S(t)\geq S^*(t)$, $\forall t\in [t^*_b, t^*_h]$, and  $\hat{\beta}_{\max}\geq\beta$, $\hat{\gamma}_{\min}\leq\gamma$, we have $\hat u(t)\geq u^*(t)$, $\forall t \in [t^*_b,t^*_h]$. 
\end{proof}
Lemma \ref{lem:inaccu} explores the case where the estimated upper and lower bounds on the parameters $\beta$ and $\gamma$ are time-invariant, and 
the states are overestimated. 
Lemma \ref{lem:inaccu} implies that $\hat u(t)=u^*(t)= \underline{u}$, $\forall t\in [0, \hat{t}_b)\cup [\hat{t}_h, +\infty)$. In addition, compared to $u^*(t)$, the proposed testing policy $\hat{u}(t)$ from Lemma \ref{lem:inaccu} starts to raise the testing rate from $\underline{u}$ earlier, and switches back to $\underline{u}$ later. Thus, to compare the cost between $\hat u(t)$ and the optimal control policy $u^*(t)$, we have the following lemma.
\begin{lemma}
\label{lem:cost}
The overall cost by leveraging $\hat{\beta}(t)=\hat{\beta}_{\max}\geq\beta$, $\hat{\gamma}(t)=\hat{\gamma}_{\min}\leq\gamma$, $\hat S(t) \in[S(t),\hat{S}_{\max}(t)]$, $\hat I(t) \in[\hat I(t),\hat{I}_{\max}(t)]$ $\forall t\geq 0$, is higher than the optimal cost by
$\int_{\hat{t}_b}^{\hat{t}_h}(\beta (S(t)-S^*(t))) dt 
-log(I(\hat{t}_h))+log(I^*(\hat{t}_h)).$
\end{lemma}
\begin{proof}
From \eqref{eq:I_u}, we have 
 $u(t)=-\frac{1}{I(t)}\frac{dI(t)}{dt}+\beta S(t) -\gamma $. By integrating the equation, as was done in \cite[Lemma 5]{acemoglu2021optimal}, and comparing $\hat{u}(t)$ and $u^*(t)$,
 we have 
$\int_{\hat{t}_b}^{\hat{t}_h} (\hat u(t)-u^*(t))dt =
\int_{\hat{t}_b}^{\hat{t}_h}(\beta (S(t)-S^*(t))) dt 
-log(I(\hat{t}_h))+log(I^*(\hat{t}_h))$, 
where $log(I(\hat{t}_b))=log(I^*(\hat{t}_b))$ is used.
\end{proof}
The difference between $\hat u(t)$ and $u^*(t)$ is captured by the difference between the susceptible states $S(t)$ and $S^*(t)$, and the infection states when the systems reach the computed herd immunity time step $\hat{t}_h$. 
Lemma \ref{lem:inaccu} and \ref{lem:cost} 
study one approach to guarantee the system's feasibility when knowing the ranges of the parameters and states. 

Now, the next result analyzes the testing strategy given in Definition \ref{def:optimal_p} with possibly time-varying estimates by leveraging the analysis from Lemma \ref{lem:inaccu} and \ref{lem:cost}.

\begin{theorem}
\label{Thm:Bound}
The testing strategy $\hat u(t)$ from Definition~\ref{def:optimal_p} by leveraging $\hat{\beta}_{\max}(t)$, $\hat{\gamma}_{\min}(t)$, $\hat{S}_{\max}(t)$ and $\hat{I}_{\max}(t)$, $\forall t\geq 0$ 
satisfies $\hat u(t)\geq u^*(t)$ $\forall t\geq 0$. Further,
the optimality gap is bounded by $\int_{\hat {t}_b}^{\hat{t}_h}(\beta (S(t)-S^*(t))) dt 
-log(I(\hat{t}_h))+log(I^*(\hat{t}_h))$, where $S(t)$ and $I(t)$ $\forall t\geq 0$ are the \bshee{true} states generated by using $\hat u(t)$.
\end{theorem}
\begin{proof}
Let $\hat S(t)$, $\hat I(t)$, $\hat u(t)$, $\forall t\geq 0$ denote the estimated states and the testing policy obtained via leveraging $\hat{\beta}_{\max}(t)$, $\hat{\gamma}_{\min}(t)$, $\hat{S}_{\max}(t)$, and $\hat{I}_{\max}(t)$, $\forall t\geq 0$ (the extremes \bshee{of} the estimated states), \bshee{while $S(t)$ and $I(t)$ $\forall t\geq 0$ denote the corresponding true states}. Let $\hat t_b$ and $\hat t_h$ denote the 
time steps when $\hat I_{\max}(t)$ reaches $\bar{I}$ for the first time and when the system reaches computed herd immunity, respectively.
From Lemma \ref{lem:inaccu}, we have that the system under the testing policy $\hat u(t)$ is feasible. Hence, we study the upper bound on the learning cost by first showing that $\hat u(t)\geq \hat u^*(t)$ $\forall t\geq 0$. Following Proposition \ref{Prop:policy} and the proof of Lemma \ref{lem:inaccu}, any feasible non-optimal control policies will consume extra tests only during the time interval $t\in[\hat {t}_{b}, \hat t_{h}]$. In addition, the overestimation of the seriousness of the epidemic by leveraging $\hat{\beta}_{\max}(t)$, $\hat{\gamma}_{\min}(t)$, $\hat{S}_{\max}(t)$, and $\hat{I}_{\max}(t)$, $\forall t\geq 0$, will cause the testing policy $\hat u(t)$ to start to switch away from $\underline{u}$ earlier but switch back to $\underline{u}$ later. Hence, we have
$\hat{u}(t)= u^*(t)=\underline{u}$ $\forall t\in [0, \hat t_b)\cup [\hat t_h,+\infty)$, and
$\hat{u}(t)\geq u^*(t)=\underline{u}$ $\forall t\in [\hat t_b, t^*_b)\cup (t^*_h, \hat t_h)$. Note that we have $S(t^*_b)\geq S^*(t^*_b)$, since $I^*(t)\geq I(t)$ $\forall t\in [0, t^*_b)$ will lead to $log (S(t))=log (S(0))-\int_{0}^{t^*_b}(\beta I(\tau)) d\tau\geq log (S(0))-\int_{0}^{t^*_b}(\beta I^*(\tau)) d\tau= log (S^*(t))$. We further compare $\hat u(t)$ and $u^*(t)$ $\forall t\in [t^*_b,t^*_h]$. Following the proof of Lemma \ref{lem:inaccu}, the optimal control strategy $u^*(t)$ maintains $I^*(t)=\bar{I}$ $\forall t\in [t^*_b, t^*_h]$, and $I(t)\leq\bar{I}=I^*(t)$,  $\forall t\in [t^*_b, t^*_h]$. Then, from
$log (S(t))=log (S(t^*_b))-\int_{t^*_b}^{t}(\beta I(\tau)) d\tau$, if $ I(t)\leq\bar{I}=I^*(t)$ $\forall t\in [t^*_b, t^*_h]$, then $\hat S_{\max}(t)\geq S(t)\geq S^*(t)$, $\forall t\in [t^*_b, t^*_h]$ (note that $S(t^*_b)\geq S^*(t^*_b)$). 
From 
the fact that $\hat S_{\max}(t)\geq S^*(t)$,  $\hat{\beta}_{\max}(t)\geq\beta$, $\hat{\gamma}_{\min}(t)\leq\gamma$, $\forall t\in [t^*_b, t^*_h]$,
and Definition \ref{def:optimal_p}, we have $\hat u(t)=\hat{\beta}_{\max}(t)\hat{S}_{\max}(t)-\hat{\gamma}_{\min}(t)\geq \beta S(t)-\gamma=u^*(t)$, $\forall t \in [t^*_b,t^*_h]$. Thus,
we have shown that $\hat u(t)\geq u^*(t)$, $\forall t\geq 0$.  From the proof of Lemma \ref{lem:cost}, by replacing $\hat I(t)$ as $\hat I_{\max}(t)$ and $\hat S(t)$ as $\hat S_{\max}(t)$, we obtain the optimality gap, which completes the proof.
\end{proof}
Theorem \ref{Thm:Bound} studies the testing strategy proposed in Definition \ref{def:optimal_p}.
Under the condition that the ranges of the learned parameters and estimated states are known, i.e., $\hat{\beta}(t)
,
\beta\in [\hat{\beta}_{\min}(t),\hat{\beta}_{\max}(t)]$; $\hat{\gamma}(t),
\gamma \in [\hat{\gamma}_{\min}(t),\hat{\gamma}_{\max}(t)]$; 
$\hat S(t),
S(t) \in[\hat{S}_{\min}(t),\hat{S}_{\max}(t)]$; $\hat I(t),
I(t) \in[\hat{I}_{\min}(t),\hat{I}_{\max}(t)]$ $\forall t\geq 0$, Definition~\ref{def:optimal_p} casts the testing by overestimating the seriousness of the epidemic at any given time step. Theorem \ref{Thm:Bound} ensures the system in \eqref{Eq: Con_Dynamics} is feasible via leveraging Definition~\ref{def:optimal_p}.
Further, Theorem \ref{Thm:Bound} provides a bound on the testing cost \baike{under uncertainties captured by} the ranges of the learned parameters and estimated states. In addition, Theorem \ref{Thm:Bound} shows that, by leveraging Definition \ref{def:optimal_p}, the susceptible state dominates the trajectory of the optimal susceptible state $\forall t\in[0, t^*_h]$,
which gives the following result.
\begin{corollary}
\label{cor:sus}
For any time $t$ up to the herd immunity time step $t^*_h$, $t\in [0,t_h^*]$, the cumulative number of people infected for the optimal testing strategy, $I^*(t)+R^*(t)$, 
will be greater than or equal to the cumulative number of people infected from the proposed testing strategy in Definition~\ref{def:optimal_p}, $I(t)+R(t)$.
\end{corollary}
From Lemma \ref{lem:inaccu}, \ref{lem:cost}, Theorem \ref{Thm:Bound}, and Corollary \ref{cor:sus}, we reach the following conclusions on the testing strategy given in Definition~\ref{def:optimal_p}.
\begin{remark}
\label{remark}
When learning and estimation strategies offer $\hat{\beta}(t)\in [\hat{\beta}_{\min}(t),\hat{\beta}_{\max}(t)]$, $\hat{\gamma}(t)\in [\hat{\gamma}_{\min}(t),\hat{\gamma}_{\max}(t)]$, $\hat S(t) \in[\hat{S}_{\min}(t),\hat{S}_{\max}(t)]$, $\hat I(t) \in[\hat{I}_{\min}(t),\hat{I}_{\max}(t)]$ $\forall t\geq 0$, compared to the optimal control strategy given in Proposition~\ref{Prop:policy},
the testing strategy from Definition \ref{def:optimal_p} under \baike{uncertainties captured by} the parameter learning and state estimation process will:

\begin{enumerate}
    \item Overestimate the seriousness of the epidemic at any given time step;

    \item React earlier to the outbreak and switch back to the lower bound on the testing rate later;

    \item Cost more or the same in terms of testing at each time $t$ $\forall t\geq0$;
    \item 
    Generate \baike{fewer or equal total uninfected individuals in the population} at any given time step up to $t^*_h$.
\end{enumerate}
\end{remark}

\section{Simulation}
\label{section4}
We now illustrate the proposed testing  strategy from Definition \ref{def:optimal_p} via simulations. Consider an epidemic spreading process in \eqref{Eq: Con_Dynamics} with $\beta=0.16$ and $\gamma =0.033$. The goal is to minimize the total number of tests 
during the epidemic given by \eqref{eq:prob} while maintaining the infection level under (or equal to) $1\%$ of the population, i.e. $\bar{I}=0.01$. We update the parameters, states, and testing policies daily, under the condition that the daily upper and lower bounds on the testing rates are $\bar{u}=15\%$ and  $\underline{u}=3\%$, respectively. 
The initial conditions are $I(0)=0.00001$, $R(0)=0$, $S(0)=1-I(0)$. The observed data sets are corrupted with noise, and the signal-to-noise ratio is $55dB$. 
From Fig.~\ref{fig_framework}, the observed data will impact both model parameter \baike{estimation} and the computation of the control input. 
We leverage the testing policies given by Proposition \ref{Prop:policy} and Definition \ref{def:optimal_p}, then compare the results.

Besides the optimal testing strategy that leverages the true parameters and states, we consider two types of testing strategies. The first testing strategy (Strategy~\bshee{1}) is to leverage Proposition~\ref{Prop:policy} by \baike{considering the noisy data and estimated parameters as the  states and model parameters for policy-making, respectively}. The second testing strategy (Strategy \bshee{2}) is to leverage Definition \ref{def:optimal_p}, where the ranges of the parameters and states are given daily. Fig.~\ref{fig:SIR} shows the comparison between the epidemic dynamics under three testing strategies, \bshee{while} the parameter estimation process via generalized linear regression \cite{draper1998applied} is shown in Fig.~\ref{fig:SIR_P}.
\bshee{From Fig.~\ref{fig:SIR_P}, we find that the transmission rate $\hat\beta(t)$ is highly underestimated during the spreading process, which may lead to the underestimation of the seriousness of the epidemic. Hence, we can compare the robustness of Strategy~1 and Strategy~2 against model uncertainties.}
Note that we use $I^*(t)$, $S^*(t)$ to represent the system trajectories under the optimal daily testing rate $u^*(t)$ and the cumulative cost $u_{total}^*(t)$. Similarly, we use $I_1(t)$ and $S_1(t)$ and $I_2(t)$, $S_2(t)$ to denote the \bshee{true} system trajectories under Strategy~\bshee{1}\baike{:} $\hat u_1(t)$ and Strategy \bshee{2}\baike{:} $\hat u_2(t)$, respectively. \bshee{Note that the corresponding noisy states ($\hat I_1(t)$, $\hat S_1(t)$, $\hat I_2(t)$, and $\hat S_2(t)$) which we leverage for parameter estimation and control policy generation are not shown in these plots.}
\begin{figure}[h]
    \centering
\includegraphics[ trim = 0cm 0cm 0cm 0cm, clip, width=\columnwidth]{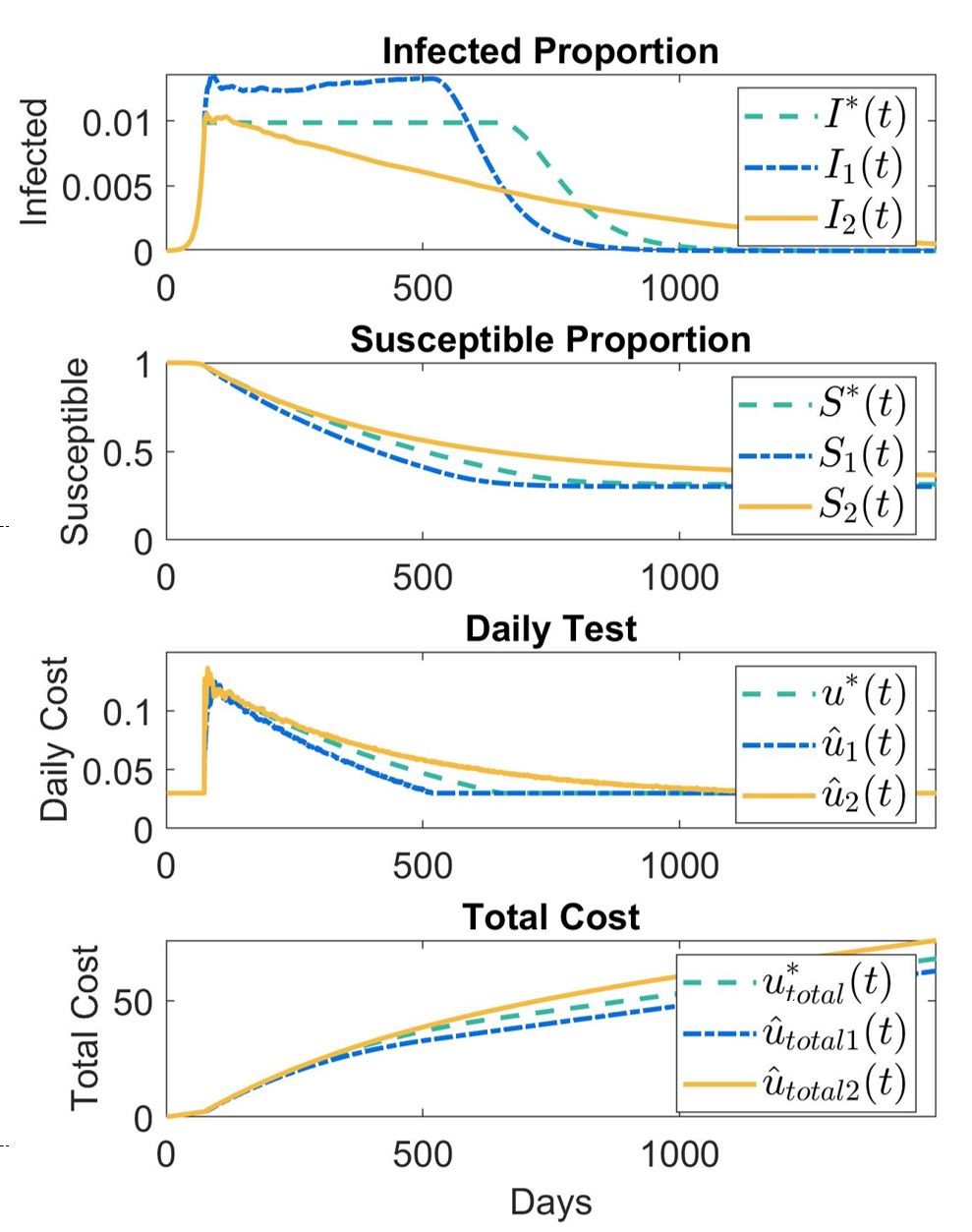}
\centering{}\caption{Comparison Between Testing Strategies}
\label{fig:SIR}
\end{figure}

\begin{figure}[h]
    \centering
\includegraphics[ trim = 0.5cm 1cm 1cm 1cm, clip, width=\columnwidth]{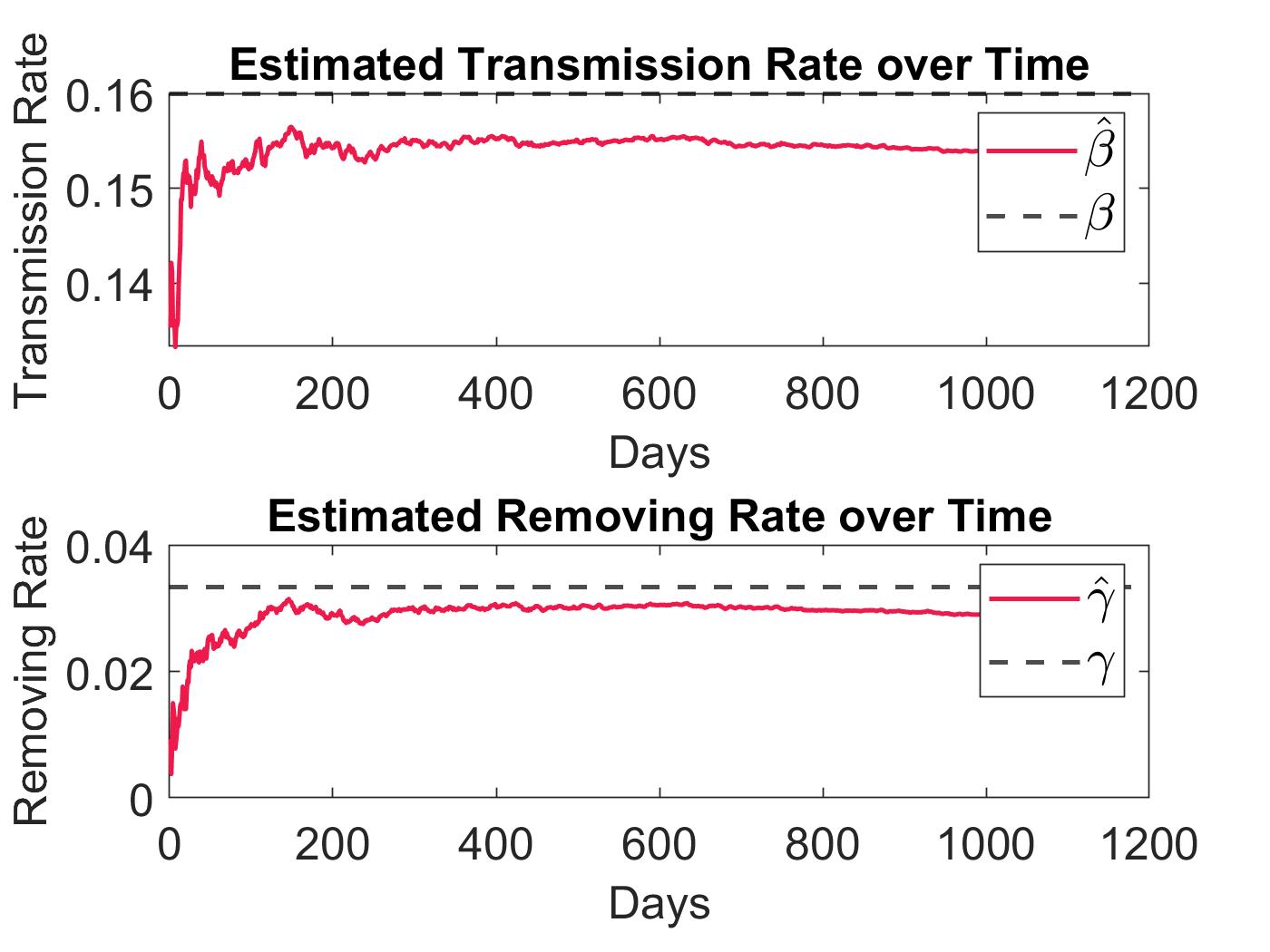}
\centering{}\caption{Parameter Estimation}
\label{fig:SIR_P}
\end{figure}
We compare the trajectories in Fig.~\ref{fig:SIR}. The simulation illustrates that the control system is \bshee{nearly} feasible by leveraging Strategy \bshee{2}, as demonstrated by $I_2(t)$ in Fig.~\ref{fig:SIR}. However, when leveraging the learned parameters directly  (Strategy~\bshee{1}), the system becomes infeasible. As shown in Fig.~\ref{fig:SIR}, the infection state $I_1(t)$ is still increasing after reaching $\bar{I}$. The cause of this phenomenon is that when the system starts to change the testing policy from $\hat{u}_1(t)=\underline{u}$ to $\hat u_1(t)=\hat{S}_1(t)\hat{\beta}(t)-\hat{\gamma}(t)$ at the time step when $\hat{I}_1(t)\geq \bar{I}$, \bshee{the highly underestimated transmission rate $\hat\beta(t)$, shown in Fig.~\ref{fig:SIR_P},} leads to the underestimation of the seriousness of the epidemic, and the testing rate $\hat u_1(t)$. Further, as illustrated in Fig.~\ref{fig:SIR}, $\hat u_1(t)\leq u^*(t)$ during the epidemic outbreak, \bshee{which will generate insufficient testing resources to maintain the infection level under the infection threshold $\bar{I}$}. Recall from Lemma \ref{lem:cost}, the optimal control policy is the pointwise smallest testing strategy we can leverage to ensure the system is feasible. Hence, $\hat u_1(t)\leq u^*(t)$ during the outbreak will lead to the system becoming infeasible. 
Regarding \baike{the second statement of Remark~\ref{remark}}, the simulation shows that it takes longer for the system under Strategy~\bshee{2} to reach the herd immunity, compared to the system under the optimal testing strategy. The daily testing generated through Strategy~\bshee{2} is higher than the optimal daily testing,
captured by $ \hat u_1(t)\geq u^*(t)$ $\forall t\geq 0$.
By comparing the simulated susceptible states, 
we see Strategy \bshee{2} generates fewer or equal total uninfected population at any given time step, i.e., $S_1(t)\geq S^*(t)$ $\forall t\geq0$, which
implies that Strategy \bshee{2} will cause fewer people to be infected over the course of the outbreak, that is, $I^*(t)+R^*(t)\geq I_1(t)+R(t)$ for all $t\geq 0$. 

\section{Conclusion}
\label{section5}
In this work, we study the impact of \baike{uncertainties introduced by}
parameter learning and state estimation \baike{in real-time} optimal epidemic mitigation.
We show the effectiveness of the proposed testing strategy \baike{when overestimating the seriousness of the epidemic under the condition that the ranges of the parameters and states are known}. Compared to the optimal testing strategy, the proposed strategy 
can flatten the curve effectively with more cost in terms of testing and time. However, we have shown analytically that the proposed strategy 
generates fewer or equal cumulative infected individuals at any given time step up to the optimal herd immunity point and it appears, via simulations, to be true for all time.
In the current work, we assume the ranges of both the states and parameters are known. Future work will propose strategies to learn the parameters and embed the parameter learning techniques into the proposed testing and isolation framework.

\normalem
\bibliographystyle{IEEEtran}
\bibliography{IEEEabrv,main}


\end{document}